\documentclass[runningheads]{llncs}
\usepackage{graphicx}

\usepackage[all]{xy}

\usepackage{stmaryrd}

\usepackage{amsmath}
\usepackage{amssymb}

\newcommand{\bis}{\mathrel{\mathchoice%
{\raisebox{.2ex}{$\,
  \underline{\makebox[.7em]{$\leftrightarrow$}}\,$}}%
{\raisebox{.2ex}{$\,
  \underline{\makebox[.7em]{$\leftrightarrow$}}\,$}}%
{\raisebox{.1ex}{$\,
  \underline{\makebox[.5em]{\scriptsize$\leftrightarrow$}}\,$}}%
{\raisebox{.1ex}{$\,
  \underline{\makebox[.5em]{\scriptsize$\leftrightarrow$}}\,$}}}}

\newcommand{\WAL}{\textsf{WAML}}
\newcommand{\PML}{\textsf{PML}}
\newcommand{\FOL}{\textsf{FOL}}

\newcommand{\M}{\mathcal{M}}
\newcommand{\F}{\mathcal{F}}
\newcommand{\N}{\mathcal{N}}

\newcommand{\la}{\langle}
\newcommand{\ra}{\rangle}
\newcommand{\lr}[1]{\la #1 \ra}

\newcommand{\pdia}{\Delta}
\newcommand{\pbox}{\nabla}
\renewcommand{\phi}{\varphi}

\newcommand{\AxWK}{\mathtt{K}}
\newcommand{\AxC}{\mathtt{C}}
\newcommand{\NecK}{\mathtt{N}}
\newcommand{\MonoK}{\mathtt{RM}}

\newcommand{\SK}{\mathbb{K}}
\newcommand{\SWK}{\mathbb{K}}

\newcommand{\Kv}{\mathsf{Kv}}

\renewcommand{\pbox}{}
\renewcommand{\pdia}{}

\newcommand{\pb}{\nabla}
\begin{document}
\title{Weakly Aggregative Modal Logic: \\ Characterization and Interpolation
%\thanks{Supported by organization x.}
}
\author{Jixin Liu\inst{1}
%\orcidID{} 
\and
Yanjing Wang\inst{2}
%\orcidID{} 
\and
Yifeng Ding\inst{3}
%\orcidID{}
}
\authorrunning{Liu, Wang, Ding}
\institute{Department of Philosophy, Peking University, China
%\email{ljx900228@163.com}
\and
Department of Philosophy, Peking University, China
%\email{wangyanjing@gmail.com}\\
%\url{http://www.phil.pku.edu.cn/personal/wangyj/index.html} 
\and
Group in Logic and the Methodology of Science, UC Berkeley, United States\\
%\email{yf.ding@berkeley.edu}\\
%\url{https://www.voidprove.com}
}
\maketitle              % typeset the header of the contribution
\begin{abstract}
Weakly Aggregative Modal Logic ($\WAL$) is a collection of disguised polyadic modal logics with n-ary modalities whose arguments are all the same. $\WAL$ has some interesting applications on epistemic logic and logic of games, so we study some basic model theoretical aspects of $\WAL$ in this paper. Specifically, we give a van Benthem-Rosen characterization theorem of $\WAL$ based on an intuitive notion of bisimulation and show that each basic $\WAL$ system $\SK_n$ lacks Craig Interpolation.

\end{abstract}
\section{Introduction}
You are invited to a dinner party for married couples after a logic conference in China. The host tells you the following facts:
\begin{itemize}
\item At least one person of each couple is a logician,
\item At least one person of each couple is Chinese.
\end{itemize}
Given these two facts, can you infer that at least one person of each couple is a Chinese logician? The answer is clearly negative, since there might be a couple consisting of a foreign  logician and a Chinese spouse who is not a logician.

Now, suppose that the host adds another fact:
\begin{itemize}
\item At least one person of each couple likes spicy food.
\end{itemize}
What do you know now? Actually, you can infer that for each couple, one of the two people must be either: 
\begin{itemize}
\item a Chinese logician, or 
\item a logician who likes spicy food, or
\item a Chinese who likes spicy food. 
\end{itemize}
This can be verified by the \textit{Pigeonhole Principle}: for each couple, there are a logician, a Chinese, and a fan for spicy food, thus there must be at least one person of the couple  who has two of those three properties. This can clearly be generalized to $n$-tuples of things w.r.t. $n+1$ properties. %Even more generally, for $km+1$ properties distributed to $m$-tuples of things, for each $m$-tuple, at least one thing has $k+1$ properties. 

Now, going back to logic, if we express ``at least one person of each couple has property $\phi$'' by $\Box \phi$ then the above reasoning shows that the following is not valid: 
$$ \AxC: \Box p\land \Box q\to \Box (p\land q).$$ On the other hand, the following should be valid: 
$$\AxWK_2:  \Box p\land \Box q\land \Box r\to \Box ((p\land q)\lor (p\land r)\lor (q\land r)).$$
In general, if $\Box \phi$ expresses ``at least one thing of each (relevant) $n$-tuple of things has property $\phi$'' then the following is intuitively valid: 
$$\AxWK_n: \Box p_{0}\wedge\cdots\land \Box p_{n}\rightarrow\Box\bigvee\limits_{(0\leq i<j\leq n)}(p_{i}\wedge p_{j}).$$
Note that $\AxWK_1$ is just $\AxC$, which is a theorem in the weakest normal modal logic $\SK$. $\AxC$ is sometimes called the \textit{Closure of Conjunction} \cite{Chellas80}, or \textit{Aggregative Axiom} \cite{jennings1981some}, or \textit{Adjunctive Axiom} \cite{Costa2005}. Clearly, when $n\geq 2$, $\AxWK_n$ are weaker versions of $\AxC$. The resulting logics departing from the basic normal modal logics by using weaker aggregative axioms $\AxWK_n$ instead of $\AxC$ are called \textit{Weakly Aggregative Modal Logics} (\WAL) \cite{schotch1980modal}. There are various readings of $\Box p$ under which it is intuitive to reject $\AxC$ besides the one we mentioned in our motivating party story. For example, if we read $\Box p$ as ``$p$ is obligatory'' as in deontic logic, then $\AxC$ is not that reasonable since one may easily face two conflicting obligations without having any single contradictory  obligation \cite{schotch1980modal}. As another example, in epistemic logic of knowing how \cite{Wang17,Kh17}, if $\Box p$ expresses ``knowing how to achieve $p$'', then it is reasonable to make $\AxC$ invalid: you may know how to get drunk and know how to prove a deep theorem without knowing how to prove it when drunk. 

Coming back to our setting where $\AxWK_n$ are valid, the readings of $\Box \phi$ in those axioms may sound complicated, but they are actually grounded in a more general picture of \textit{Polyadic Modal Logic ($\PML$)} which studies the logics with $n$-ary modalities. Polyadic modalities arose naturally in the literature of philosophical logic, particularly for the binary ones, such as the \textit{until} modality in temporal logic \cite{Kamp1968}, instantial operators in games-related neighborhood modal logics \cite{van2017instantial}, relativized knowledge operators in epistemic logic \cite{LCC,WF14}, and the conditional operators in the logics of conditionals \cite{Beall2012}. Following the notation in \cite{blackburn2002modal}, we use $\nabla$ for the $n$-ary generalization of the $\Box$ modality when $n>1$.\footnote{This is not to be confused with the non-contingency operator, which is also denoted as $\nabla$ in non-contingency or knowing whether logics \cite{FWvD15}.} The semantics of $\nabla(\phi_1, \dots, \phi_n)$ is based on Kripke models with $n+1$-ary relations $R$ \cite{jennings1981some,blackburn2002modal}: 
%\noteYW{Who came up with the semantics? Jennings?} 

\begin{quote}$\nabla(\phi_1, \dots, \phi_n)$ holds at $s$ iff \textit{for all} $s_1, ..., s_n$ such that $Rss_1\dots s_n$ \textit{there exists} some $i\in[1,n]$ such that $\phi_i$ holds at $s_i$.
\end{quote}
Note that the quantifier alternation pattern $\forall\exists$ in the above (informal) semantics for $\nabla$. Actually, the reading we mentioned for $\Box\phi$ in our motivating story is simply the semantics for $\nabla(\phi_1, \dots, \phi_n)$ where $\phi_1=\dots=\phi_n$. Essentially, the formulas $\Box\phi $ under the new reading can be viewed as special cases of the modal formulas in polyadic modal languages. Due to the fact that the arguments are the same in $\nabla(\phi, \dots, \phi)$, we can also call the $\Box$ under the new reading the \textit{diagonal $n$-modalities}.\footnote{Name mentioned by Yde Venema via personal communications.} In this light, we may call the new semantics for $\Box\phi$ the \textit{diagonal $n$-semantics} (given frames with $n+1$-ary relations).  

Diagonal modalities also arise in other settings in disguise. For example, in epistemic logic of \textit{knowing value} \cite{GW16}, the formula $\Kv (\phi, c)$ says that the agent knows the value of $c$ given $\phi$, which semantically amounts to that for all the pairs of $\phi$ worlds that the agent cannot distinguish from the actual worlds, $c$ has the same value. In other words, in \textit{every} pair of the indistinguishable worlds where $c$ has \textit{different} values, \textit{there is} a $\neg \phi$ world, which can be expressed by $\Box^c \neg \phi$ with the diagonal 2-modality ($\Box^c$) based on intuitive ternary relations (see details in \cite{GW16}). As another example in epistemic logic, \cite{RAK} proposed a local reasoning operator based on models where each agent on each world may have different \textit{frames of mind} (sets of indistinguishable worlds). One agent believes $\phi$ then means that in \textit{one} of his current frame of mind, $\phi$ is true \textit{everywhere}. This belief  modality can also be viewed as the dual of a diagonal 2-modality (noticing the quantifier alternation $\exists \forall$ in the informal semantics). 

Yet another important reason to study diagonal modalities comes from the connection with paraconsistent reasoning established by Schotch and Jennings \cite{schotch1980modal}. In a nutshell, \cite{schotch1980modal} introduces a notion of \textit{$n$-forcing} where a set of formulas $\Gamma$ $n$-forces $\phi$ ($\Gamma\vdash_n\phi$) if for each $n$-partition of $\Gamma$ there is a cell $\Delta$ such that $\phi$ follows from $\Delta$ classically w.r.t.\ some given logic ($\Gamma\vdash \phi$). This leads to a notion of \textit{$n$-coherence} relaxing the notion of consistency:  $\Gamma\nvdash_n \bot$ ($\Gamma$ is $n$-coherent) iff there exists an $n$-partition of $\Gamma$ such that all the cells are classically consistent. These notions led the authors of \cite{schotch1980modal} to the discovery of the diagonal semantics for $\Box$ based on frames with $n+1$-ary relations, by requiring $\Box(u)=\{\phi\mid u\vDash \Box \phi\}$ to be an $n$-theory based on the closure over $n$-forcing, under some other minor conditions. Since the derivation relation of basic normal modal logic $\SK$ can be characterized by a proof system extending the propositional one with the rule $\Gamma\vdash \phi \slash \Box(\Gamma)\vdash \Box\phi$ where $\Box(\Gamma)=\{\Box\phi\mid \phi\in\Gamma\}$, it is interesting to ask whether adding $\Gamma\vdash_{n} \phi \slash \Box(\Gamma)\vdash \Box \phi$ characterizes exactly the valid consequences for modal logic under the diagonal semantics based on frames with $n$-ary relations. Apostoli and Brown answered this question positively in \cite{apostoli1995solution} 15 years later, and they characterize $\vdash_n$ by a Gentzen-style sequent calculus based on the compactness of $\vdash_n$ proved by using a compact result for coloring hypergraphs.\footnote{Other connections between \WAL\ and graph coloring problems can be found in \cite{NA02} where the four-color problem is coded by the validity of some formulas in the \WAL\ language.} Moreover, they show that the \WAL\ proof systems with $\AxWK_n$ are also complete w.r.t. the class of all frames with $n$-ary relations respectively. The latter proof is then simplified in \cite{nicholson2000revisiting} without using the graph theoretical compactness result. This completeness result is further generalized to the extensions of \WAL\ with extra one-degree axioms in \cite{apostoli1997completeness}. The computational complexity issues of such logics are discussed in \cite{allen2005complexity}, and this concludes our relatively long introduction to \WAL, which might not be that well-known to many modal logicians.  
%\noteYW{We can probably say sth about the use of the language in stating colorability in graph theory.}

In this paper, we continue the line of work on \WAL\ by looking at the model theoretical aspects. In particular, we mainly focus on the following two questions: 
\begin{itemize}
\item How to characterize the expressive power of \WAL\ structurally within first-order logic over (finite) pointed models? 
\item Whether \WAL\ has Craig Interpolation?
\end{itemize}
For the first question, we propose a notion of bisimulation to characterize \WAL\ within the corresponding first-order logic. The answer for the second question is negative, and we will provide counterexamples in this paper to show \WAL\ do not have Craig Interpolation.

In the rest of the paper, we lay out the basics of \WAL\ in Section \ref{sec.pre}, prove the characterization theorem based on a bisimulation notion in Section \ref{sec.bis}, and give counterexamples for the interpolation theorem in Section \ref{sec.interp} 
% and also give some partial results about interpolation in section \ref{sec.interp} 
before concluding with future work in \ref{sec.con}. 

\section{Preliminaries}
\label{sec.pre}

In this section we review some basic definitions and results in the literature.

\subsection{Weakly Aggregative Modal Logic}
The language for $\WAL$ is the same as the language for basic (monadic) modal logic.
\begin{definition}
Given a set of propositional letters $\Phi$ and a single unary modality $\Box$, the language of $\WAL$ is defined by:
$$\phi:=p\mid\lnot\phi\mid(\phi\wedge\phi)\mid\Box\phi$$
where $p\in\Phi$. We define $\top$, $\phi\lor\psi$, $\phi\to\psi$, and $\Diamond \phi$ as usual. 
\end{definition}

However, given $n$, $\WAL$ can be viewed as a fragment of polyadic modal logic with a $n$-ary modality, since $\Box\phi$ is essentially $\pb(\phi, \ldots, \phi)$. \textbf{Notation:} in the sequel, we use $\WAL^n$, where $n>1$, to denote the logical framework with the semantics based on $n$-models defined below:
\begin{definition}[$n$-Semantics]
An $n$-frame is a pair $\lr{W, R_{\pbox}}$ where $W$ is an nonempty set and $R_{\pbox}$ is an $n+1$-ary relation over $W$. A $n$-model $\M$ is a pair $\lr{\F, V}$ where the valuation function $V$ assigns each $w\in W$ a subset of $\Phi$. We say $\M$ is an \textit{image-finite} model if there are only finitely many $n$-ary successors of each point. The semantics for $\Box \phi$ (and $\Diamond \phi$) is defined by: 
\begin{small}$$\begin{array}{|lll|}
\hline
\M,w\models\Box \phi  & \text{iff} & \text{ for all $v_{1},\ldots v_{n}\in W$ with  $R_{\pbox}wv_{1}\ldots,v_{n}$}, \M,v_{i}\models\phi \text{ for some $i\leq n$}.\\
 \M,w\models\Diamond \phi  & \text{iff} & \text{ there are $v_{1},\ldots v_{n}\in W$ st. $R_{\pbox}wv_{1}\ldots,v_{n}$ and } \M,v_{i}\models\phi \text{ for all $i\leq n$}.\\
 \hline
\end{array}
$$
\end{small}
\end{definition}
According to the above semantics, it is not hard to see that the aggregation axiom $\Box\phi\land \Box\psi\to \Box(\phi\land \psi)$ in basic normal modal logic is not valid on $n$-frames for any $n > 1$. 
%For example, $\Box p\land \Box q\to \Box(p\land q)$ does not hold at $w$ in the following 2-model where the triangle denote the ternary relation: 
%$$\xymatrix@R-20pt@C+10pt{ 
%                         & v:p\ar@{-}[dd]\\
%w\ar@{-}[ru]\ar@{-}[rd]&\\
 %                       & v:q
 %}$$

% Instead of the aggregation axiom, we have some weaker ones which are valid. For example, given $k=2$, $(\Box p\land \Box q\land \Box r) \to \Box((p\land q)\lor(p\land r)\lor (q\land r))$ is valid over $2$-frames which can be proved by a simple \textit{pigeon hole} argument. 
\cite{schotch1980modal} proposed the following proof systems $\SK_n$ for each $n$. 
\begin{definition}[Weakly aggregative modal logic]
The logic $\SK_n$ is a modal logic including propositional tautologies, the axiom $\AxWK_n$ and closed under the rules $\NecK$ and $\MonoK$: 
$$
\begin{array}[c]{cc}
\AxWK_n & \Box p_{0}\wedge\cdots\land \Box p_{n}\rightarrow\Box
\bigvee\limits_{(0\leq i<j\leq n)}(p_{i}\wedge p_{j})\\
\NecK &\vdash\phi  \implies  \vdash\Box\phi\\
\MonoK     & \vdash\phi\rightarrow\psi \implies \vdash\Box\phi
\rightarrow\Box\psi\\
%\RS & \vdash \phi (p) \implies \vdash \phi[\psi\slash p] 
\end{array}
$$
\end{definition}
It is clear that $\AxWK_1$ is just the aggregation axiom $\AxC$ and thus $\SWK_1$ is just the normal monadic modal logic $\SK$. It can also be shown easily that for each $n > m$, $\SK_n$ is strictly weaker than $\SK_m$. In fact, many familiar equivalences in normal modal logics, like the equivalence between $\Diamond \top$ and $ \Box p \to \Diamond p$, no longer hold in $\SK_n$ for $n > 1$. Semantically speaking, while $\Box p \to \Diamond p$'s validity corresponds to seriality on $1$-frames (usual Kripke frames), its correspondence on $2$-frames is not even elementary ($\Diamond \top$ still corresponds to each point having at least a successor tuple). 

%Note that, to make sure $\AxC$ is valid over 3-frames, we need extra frame condition $\forall x\forall y \forall z(Rxyz \to (y=z \lor Rxyy\lor Rxzz))$ \cite{jennings1984preservation}. Also note that the diagonal modalities can distinguish equivalent formulas in the standard modal logic setting, e.g., $\neg \Box \bot$ is no longer equivalent to $\Box p\to \Diamond p$ any more w.r.t. the diagonal semantics. In fact, the property corresponding to $\Box p\to \Diamond p$ w.r.t. the validity of diagonal semantics is not first-order definable. 

After being open for more than a decade, the completeness for $\SK_{n}$ over $n$-models was finally proved in \cite{apostoli1995solution} and \cite{apostoli1997completeness}, by reducing to the $n$-forcing relation proposed in \cite{schotch1980modal}. In \cite{nicholson2000revisiting}, a more direct completeness proof is given using some non-trivial combinatorial analysis to derive a crucial theorem of $\SK_n$. 

\section{Characterization via bisimulation}
\label{sec.bis}
In this section, we introduce a notion of bisimulation for $\WAL$ and prove the van-Benthem-Rosen Characteristic Theorem for \WAL.\footnote{We have another proof for the Characterization theorem over arbitrary $n$-models, using tailored notions of saturation and ultrafilter extension for $\WAL^n$, due to the space limit we only present the proof which also works for finite models.}
%Throughout this section we fix an $n$ for logic of $\WAL$ over $n$-models. \noteYW{Mention alternative proof method.}
%\noteYW{What about restricting to models with certain properties and use the standard bisimulation. }

\begin{definition}
[$wa^n$-bisimulation] Let $\M=(W,R_{\pbox},V)$ and $\M'=(W',R_{\pbox}',V')$ be two $n$-models. A non-empty binary relation $Z\subseteq W\times W'$ is called a \emph{$wa^n$-bisimulation} between $\M$ and $\M'$ if the following conditions are satisfied:
\begin{itemize}
\item[inv] If $wZw'$, then $w$ and $w'$ satisfy the same propositional letters (in $\Phi$).
\item[forth] If $wZw'$ and $R_{\pbox}wv_{1},\ldots,v_{n}$ then there
are $v_{1}',\ldots,v_{n}'$ in $W'$ s.t. $R'_{\pbox}w'v_{1}',\ldots,v_{n}'$ and for each $v'_j$ there is a $v_i$ such that $v_i Z v'_j$ where $1\leq i,j\leq n$.
\item[back] If $wZw'$ and $R'_{\pbox}w'v_{1}',\ldots,v_{n}'$ then there are $v_{1},\ldots,v_{n}$ in $W$ s.t.\ $R_{\pbox}wv_{1},\ldots,v_{n}$ and for each $v_i$ there is a $v'_j$ such that $v_i Z v'_j$ where $1\leq i,j\leq n$.
\end{itemize}
When $Z$ is a bisimulation linking two states $w$ in $\M$ and $w'$
in $\M'$ we say that $w$ and $w'$ are $\Phi$-$wa^n$-\textit{bisimilar} $(\M,w \bis^n \M',w')$.
%If there is some bisimulation between $\M$ and $\M'$, we write $\M \bis^n \M'$, saying $\M$ and $\M'$ are bisimilar.
\end{definition}
\begin{remark}Pay attention to the two subtleties in the above definition: $i, j$ in the forth and back conditions are not necessarily the same, thus we may not have an aligned correspondence of each $v_i$ and $v_i'$; in the second part of the forth condition, we require each $v_j'$ to have a corresponding $v_i$, not the other way around. Similar in the back condition. This reflects the quantifier alternation in the semantics of $\Box$ in $\WAL^n$. 
\end{remark}
% From the definition above we can see that if $w$ and $w'$ are modal bisimilar, they are clearly \WAL-bisimilar. Since we know that \WAL\ is a fragment of polyadic modal logic, we will show that it is exactly the fragment closed under the bisimulation above. First we show the bisimulation is indeed sound w.r.t. the \WAL-equivalence.
\begin{example} Consider the following two $2$-models where $\{\lr{w,w_1,w_2}, \lr{w,w_2,w_3}\}$ is the ternary relation in the left model, and $\{\lr{v,v_1,v_2}\}$ is the ternary relation in the right model. 
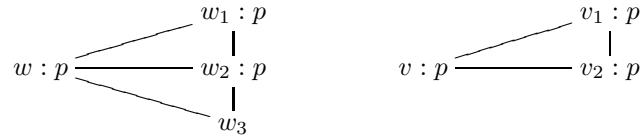
\vspace{-5pt}
$$\xymatrix@C+20pt@R-15pt{
 &w_1:p\ar@{-}[d] &     & v_1:p\ar@{-}[d] \\
w:p\ar@{-}[ru]\ar@{-}[rd]\ar@{-}[r]&w_2:p\ar@{-}[d] &  v:p\ar@{-}[ru]\ar@{-}[r]  & v_2:p \\
 & w_3      & 
}$$
$Z=\{\lr{w,v}, \lr{w_1, v_1}, \lr{w_2, v_2}, \lr{w_2, v_1}\}$ is a $wa^2$-bisimulation. A polyadic modal formula $\neg\nabla\neg(p, \neg p)$, not expressible in $\WAL^2$, can distinguish $w$ and $v$.  
\end{example}

It is easy to verify that $\bis^n$ is indeed an equivalence relation and we show $\WAL^n$ is invariant under it. 
\begin{proposition}
\label{prop:bisim-invariance}
Let $\M=(W,R_{\pbox},V)$ and $\M'=(W',R_{\pdia}',V')$ be two $n$-models. Then for every $w\in W$ and $w'\in W'$, $w \bis^n w'$ implies $w\equiv_{\WAL^n} w'$. In words, $\WAL^n$ formulas are invariant under $wa^n$-bisimulation.
\end{proposition}
\begin{proof}
We consider only the modality case. Suppose that $w \bis^n w'$
and $w\models\Diamond\phi$. Then there are $v_{1},\ldots,v_{n}$
s.t. $R_{\pbox}wv_{1},\ldots,v_{n}$, and each $v_{i}\models$ $\phi$. By
the forth condition, there are $v_{1}',\ldots,v_{n}'$ in
$W'$ s.t. $R_{\pbox}w'v_{1}',\ldots
,v_{n}'$ and and for each $v'_j$ there is a $v_i$ such that $v_i Z v'_j$.
From the I.H. we have each $v_{i}'\models\phi$. As a
result, $w'\models\Diamond\phi$. For the converse direction just
use the back condition.
\end{proof}

\begin{theorem}
[Hennessy-Milner Theorem for $\WAL^n$]Let $\M=(W,R_{\pbox},V)$ and $\M'=(W',R_{\pbox}',V')$ be two image-finite $n$-models. Then for every $w\in W$ and $w'\in W'$,
$w \bis^n w'$ iff $w\equiv_{\WAL^n} w'$.
\end{theorem}
\begin{proof} As in basic modal logic, the crucial part is to show $\equiv_{\WAL^n}$ is indeed a $wa^n$-bisimulation and we only verify the forth condition. Suppose towards contradiction that $Rwv_1\dots v_n$ but for each $v_1'\dots v'_n$ such that $R'w'v_1'\dots v'_n$ there is a $v'_j$ such that it is not $\WAL^n$-equivalent to any of $v_i$. In image-finite models we can list such $v'_j$ as $u_1\dots u_m$. Now for each $u_k$ and $v_i$ we have $\phi_k^i$ which holds on $v_i$ but not on $u_k$. Now we consider the formula $\psi=\Diamond (\bigvee_{1\leq i\leq n}\bigwedge_{1\leq k\leq m} \phi_k^i)$. It is not hard to see that $\psi$ holds on $w$ but not $w'$, hence contradiction.
\end{proof}
%\noteYW{Do we really use it?}
Like in normal modal logic, we can also define a notion of $k$-bisimulation of $\WAL^n$, by restricting the maximal depth we may go to.
\begin{definition}
[$k$-$wa^n$-bisimulation] Let $\M=(W,R_{\pbox},V)$ and $\M'=(W',R_{\pbox}',V')$ be two $n$-models. $w$ and $w'$ are $0$-$wa^n$-bisimilar ($w\bis^n_{0} w'$) iff $V(v)=V'(v')$. $w\bis^n_{k+1} w'$ iff $w\bis^n_{k} w'$ and the follow two conditions are satisfied: 
\begin{itemize}
\item[forth] If $v\bis^n_{k+1}v'$ and $R_{\pbox}vv_{1},\ldots,v_{n}$ then there
are $v_{1}',\ldots,v_{n}'$ in $W'$ s.t. $R'_{\pbox}v'v_{1}',\ldots,v_{n}'$ and for each $v'_j$ there is a $v_i$ such that $v_i \bis^n_k v'_j$ where $1\leq i,j\leq n$.
\item[back] If $v\bis^n_{k+1}v'$ and $R'_{\pbox}v'v_{1}',\ldots,v_{n}'$ then there are $v_{1},\ldots,v_{n}$ in $W$ s.t. $R_{\pbox}vv_{1},\ldots,v_{n}$ and for each $v_i$ there is a $v'_j$ such that $v_i \bis^n_k v'_j$ where $1\leq i,j\leq n$.
\end{itemize}
\end{definition}
% Given finitely many propositional letters, there are essentially only finitely many logically different formulas in $\WAL^n_k$. By a simple induction proof we can show: 
% \begin{proposition}
% Suppose there are only finitely propositional letters, then given two $n$-models $\M=(W,R_{\pbox},V)$ and $\M'=(W',R_{\pbox}',V')$ the following are equivalent.
% \begin{itemize}
% \item[i] $w\bis^n_{k} w'$
% \item[ii] $w$ and $w'$ agree on all $\WAL^n$ formulas of degree at most $k$.
% \end{itemize}
% \end{proposition}
% \begin{proof}
% By an induction on k with a similar strategy for proving the Hennessy-Milner theorem.
% \end{proof}
We can translate each $\WAL^n$ formula to an equivalent $\FOL$ formula with one free variable and $n+1$-ary relation symbols, thus $\WAL^n$ is also compact. 
\begin{definition}[Standard translation]
$\vspace{-5pt}
ST:\WAL^n\to \FOL$:
\[\begin{array}{lcl}
ST_x(p)&=&Px\\ 
ST_x(\neg\phi)&=& \neg ST_x(\phi)\\ 
ST_x(\phi\land\psi)&=& ST_x(\phi)\land ST_x(\psi) \\  
ST_x(\Box \phi)&=& \forall y_1\forall y_2\dots \forall y_n (R xy_1y_2\dots y_n\to ST_{y_1}(\phi)\lor \dots\lor ST_{y_n}(\phi))
\end{array}\]
\vspace{-10pt}
\end{definition}
%It is obvious that we have a standard translation from \WAL \ to $\PML$. We can just translate $\Diamond\phi$ into $\pdia(\phi,\ldots,\phi)$, where the number of arguments is depended on which polyadic language we use. 
By following a similar strategy as in \cite{otto2004elementary}, we will show a van Benthem-Rosen characterization theorem for $\WAL^n$: a $\FOL$ formula is equivalent to the translation of a $\WAL^n$ formula (over finite $n$-models) if and only if it is invariant under $wa^n$-bisimulations (over finite $n$-models).

First we need to define a notion of \textit{unraveling} w.r.t.\ $n$-ary models as we did for models with binary relations.   We use an example of a graph with ternary relations to illustrate the intuitive idea behind the general $n$-ary unraveling, which is first introduced in \cite{de1993extending}.

\begin{example}\label{ex.unr}
Given the $2$-model with ternary relations $\lr{\{w,v,u,t\},$ $\{\lr{w, u,t},$ $\lr{u,t,u},$ $\lr{t,w,v} \}, V}$. It is quite intuitive to first unravel it into a tree with pairs of states as nodes, illustrated below: 
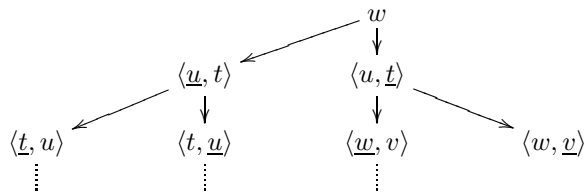
\vspace{-10pt}
$$\xymatrix@R-14pt@C+12pt{
& &w\ar[d]\ar[dl] &\\
 &\lr{\underline{u},t}\ar[ld]\ar[d] &\lr{u,\underline{t}}\ar[d]\ar[dr] &\\
 \lr{\underline{t}, u}\ar@{.}[d]&\lr{t,\underline{u}} \ar@{.}[d]&\lr{\underline{w},v} \ar@{.}[d]& \lr{w,\underline{v}}\\
 &&&
} \vspace{-10pt}$$
To turn it into a 2-model, we need to define the new ternary relations. For each triple $\lr{s_0,s_1,s_2}$ of pairs, $\lr{s_0, s_1, s_2}$ is in the new ternary relation iff $s_1$ and $s_2$ are successors of $s_0$ in the above graph and the triple of underlined worlds in $s_0, s_1, s_2$ respectively is in the original ternary relation, e.g., $\lr{u,\underline{t}},\lr{\underline{w},v},\lr{w,\underline{v}}$ is in the new ternary relation since $\lr{t,w,v}$ is in the original ternary relation. 
\end{example}
In general, we can use the $n$-tuples of the states in the original model together with a natural number $k\in [1,n]$ as the basic building blocks for the unraveling of an $n$-model, e.g., $\lr{w, v, u, 2}$ means the second the state is the underlined one. To make the definition uniform, we define the root as the sequence $\lr{w,\dots,w,1}.$ Like the unraveling for a binary graph, formally we will use sequences of such building blocks as the nodes in the unraveling of a $n$-model, e.g., the left-most node $\lr{\underline{t},u}$ in the above example will become $\lr{\lr{w,w,1},\lr{u,t,1},\lr{t, u ,1}}$. This leads to the following definition. 
\begin{definition}
Given an $n$-model $\M=\langle W,R,V\rangle$ and 
$w\in W$, we first define the binary unraveling $\M^b_{w}$ of $\M$ around $w$ as $\lr{W_w, R^b, V'}$ where: 
\begin{itemize}
\item $W_w$ is the set of sequences $\langle\langle\vec{v}_0,i_0\rangle,\langle\vec{v}_{1},i_{1}\rangle,\ldots,\langle\vec{v}_{m},i_{m}\rangle\rangle$  where:
\begin{itemize}
\item $m\in\mathbb{N}$;
\item for each $j\in [0,m]$, $\vec
{v}_{j}\in W^n$ and $i_{j}\in
[1, n]$ such that $R(\vec{v}_{j}[i_{j}])\vec{v}_{j+1}$;
\item $\vec{v}_0$ is the constant $n$-sequence $\langle w,\dots ,w\rangle $ and $i_0=1$;
\end{itemize}
\item $R^bss'$ iff $s'$ extends $s$ with some $\lr{\vec{v}, i}$
\item $V'(s)=V(r(s))$, where $r(s)= \vec{v}_m[i_m]$ if $s=\lr{\dots, \lr{\vec{v}_m,i_m}}$.
\end{itemize}
The unraveling $\M_w=\lr{W_w, R', V'}$ is based on $ \M^b_w$ by defining $R's_0 s_1\dots s_n $ iff $Rr(s_0)r(s_1)\dots r(s_n)$ and $R^bs_0s_i$ for all $i\in [1,n]$.
% \item $R's_0 s_1\dots s_n $ iff for each $i\in [1,n]$, $s_i$ extends $s_0$ with some $\lr{\vec{v}, i}$, and $Rr(s_0)r(s_1)\dots r(s_n)$ where $r(s)= \vec{v}_n[i_n]$ if $s=\lr{\dots, \lr{\vec{v}_n,i_n}}$. 
%We refer to the frame $\lr{W_w,R'}$ of the unraveling as  $\F_w$.
Let the bounded unraveling $\M_w|_l$ be the submodel of $\M_w$ up to level $l$.  
% \bigskip For each sequence $s$, let $l(s)$ be the lenth of $s$ and define
% $t(s)=s[l(s)-1]$.
% Let $r(s)=(t(s)[0])[t(s)[1]]$, which means $s$ represents a member $r(s)$ in
% $W$.
% For $s\in W^{w}$, we call $s$ a generated sequence from $w$ to $r(s)$.
% We define a relation $R^{\ast}$ between sequences: $aR^{\ast}b$ iff $b$ is a
% extension of $a$ and $l(b)=l(a)+1$;
% Now we give the definition of $R^{w}$.
% \bigskip
% $R^{w}:=\{(s,s_{1},\ldots,s_{n})\mid$ $sR^{\ast}s_{i}$ for each $i$,
% and $r(s)Rr(s_{1})\ldots r(s_{n})\}$;
% \bigskip
% $F^{w}:=\langle W^{w},R^{w}\rangle$.
\end{definition}
\begin{remark}\label{rem.tree}
Clearly $\M^b_w$ is a tree, and in $\M_w$, if $Rs_0\dots s_n$ then $s_1\dots s_n$ are at the next ``level'' of $s_0$. Such properties are crucial in the later proofs.
\end{remark}
% So by definition, there is a natural morphism $f:$ $F^{w}\rightarrow F$ by
% $f(s)=r(s)$.
$r$ defined above reveals the corresponding state of $s$ in the original model $\M$. It is not hard to show the following. 
\begin{proposition}\label{prop.bm}
The above $r$ (viewed as a relation) is a $wa^n$-bisimulation between $\M_{w}$ and $\M$. Actually  $r$ is a p-morphism (over $n$-models) from $\M_w$ to $\M$.
\end{proposition}

Now we have all the ingredients to prove the following characterization theorem. Note that the characterization works with or without the finite model constraints.
\begin{theorem}\label{thm312}
A first-order formula $\alpha(x)$ is invariant under $\bis^n$ (over  finite models) iff $\alpha(x)$ is equivalent to a $\WAL^n$ formula (over finite models). 
\end{theorem} 

Following the general strategy in \cite{otto2004elementary}, the only non-trivial part is to show that the FOL formula $\alpha(x)$ that is invariant under $wa^n$-bisimulation has some locality property w.r.t. its bounded unraveling $\M_w|_l$ for some $l$. Due to lack of space, we only show the following lemma and give a proof sketch here.
% in Appendix \ref{app.local}. 
For other relatively routine parts of the proof, see \cite{otto2004elementary}.
\begin{lemma}[locality]\label{lem.local}
An FOL formula $\alpha(x)$ is invariant under $\bis$ (over finite models) implies that for some $l \in \mathbb{N} $, for any $n$-model $\M,w$: $\M,w\Vdash\alpha(x)[w]$ iff $\M_w|_{l}\Vdash\alpha(x)[(\lr{\vec{w},1)}]$. 
\end{lemma}
Here we explain the most important ideas behind the proof. First of all, like in \cite{otto2004elementary}, we take $l=2^q-1$ where $q$ is the quantifier rank of $\alpha(x)$, and build two bigger models which are $wa^n$-bisimilar to $\M,w$ and $\M_w|_{l}$ respectively using our new unraveling notion. Then we show in the $q$-round EF game between the bigger $n$-models Duplicator has a winning strategy. To specify the strategy, which is essentially to let the duplicator to keep some ``safe zones'' for extensions of partial isomorphisms, we need to define the distance of points in $n$-models. Let the distance between $s$ and $s'$ (notation $d(s,s')$) be the length of the shortest (undirected) path between $s$ and $s'$ via a new relation binary $R^{c}$ s.t. $R^{c}xy$ iff $Rxy_{1}\ldots y_{n}$ and $y=y_{i}$ for some $i\in[1, n]$. We set  $d(s, s')=\omega$ if $s$ and $s'$ are not connected by any such path. It is easy to see that in the unraveling $\M_w$, $d(s, s')$ is exactly the distance in the usual sense between $s, s'$ in the tree $\M_w^b$. Surprisingly, the winning strategy looks exactly like the one in \cite{otto2004elementary} for binary models, this of course deserves some explanations below.

Another key point is that we need to define two ``neighborhoods'' of a node--a big one and a small one, as in the following key step in our proof:

We need to show that: After $m$ rounds ($0\leq m\leq q$), the following two hold:

(Let $(a_{i},b_{i})$ be the pair selected at $i$ round where each $a_{i}\in
\M^{\ast}$ and $b_{i}\in \N^{\ast}$, especially $a_{0}=w^{\ast}$ and
$b_{0}=v^{\ast}$.

Let $S(m)=\{a_{i}\mid i\leq m\}$, $N_{i}(m)$ be the neighborhood of $a_{i}$
within distance of $2^{q-m}-1$, and $N_{i}^{^{\prime}}(m)$ be the neighborhood
of $a_{i}$ within distance of $2^{q-(m+1)}$.)

(1). the selected points form a partial isomorphism $I$: $\M^{\ast}\rightarrow
\N^{\ast}$.

(2). if $m<q$ then there is a sequence $(I_{0},\ldots,I_{m})$ s.t. for each
$i\leq m$,

a). $I_{i}\supseteq I$ is a partial isomorphism with $Dom(I_{i})=N_{i}(m)\cup
S(m)$;

b). $\forall h,j\leq m\forall x\in N_{h}^{^{\prime}}(m)\cap N_{j}^{^{\prime}}(m)(I_{h}(x)=I_{j}(x))$.

In Otto's oringinal proof in \cite{otto2004elementary}, he omits the above part. But we think it's necessary to give such an explicit description here. 

\begin{remark}\label{require for distance 1}
It is not hard to show that under our distance notion, for each $x,y,z$ in the model, $d(x,z) \geq d(x,y)-d(y,z)$, i.e., $d(x,z)+d(z,y) \geq d(x,y)$ which is a more usual form of the  \textit{triangle inequality}. This justifies the new distance notion. To see why a similar strategy like the one in \cite{otto2004elementary} for binary models works, note that our unraveling $\M_w$ is essentially based on a \textit{tree} $\M_w^b$ by definition, and the $n$-ary relation over such a tree structure has very a special property: if $Rs_0\dots s_n$ then $s_1\dots s_n$ are immediate successors of $s_0$ in the binary unraveling as mentioned in Remark \ref{rem.tree}. This leads to the following crucial property we will use repeatedly: if we already established a partial isomorphism $I$ between $S$ and $N$ (w.r.t.\ also $n$-ary relations), and $x\not\in S$ is not directly connected to anything in $S$, and $y\not\in N$ is also not directly connected to anything in $N$ then $I\cup\{(x,y)\}$ extending $I$ is again a partial isomorphism. 
\end{remark}

% \begin{proof}
%  However, in a $n$-ary model we have to redefine the distance:  The idea is that Defender has to keep a safe matching zone of $2^{q-m}$ at round $m$. If Spoiler selects some point outside the alarm distance then Defender can respond with the ``same'' point in a new copy. Note that the winning strategy was first proposed in \cite{otto2004elementary} for the binary models in the standard modal logic setting. However, it also works in our $n$-ary models because the way we  The detailed proof is a bit complicated which can be found in the Appendix. 
% % \end{proof}
% % The following is the key part for proving the above lemma.
% % \begin{proposition}
% % $M^*,w^*$ $\equiv_{q}N^{\ast},v^{\ast}$.
% % \end{proposition}
% % \begin{proof}
% \end{proof}

Finally, the bound $l=2^{q}-1$ in the above proof, which we choose uniformly for every $n$, is actually not ``optimal", since for a larger $n$, we can have a lower bound. Especially, when $n > q$, even $l=1$, the Duplicator could have a winning strategy, since any bijection will be a partial isomorphism. So the distance we define here is not a appropriate one for us to find the minimal bound $l$. Here we conjecture that the bound should be the least integer $l$ s.t. $l \geq (2^{q}-1)/n$.

\section{Interpolation}
\label{sec.interp}
% By a standard strategy in \cite{hansen2009neighbourhood}, we know that the basic polyadic modal logics have the Craig Interpolation theorem. But we cannot use the same
% tool to prove interpolation theorem for \WAL, since the bisimulation notion for
% \WAL \ is not strong enough. Even though we have already known that \WAL\  is complete, which means $\PML$ is a conservative extension of \WAL, we cannot get the interpolation theorem directly, since we can only find a interpolation in $\PML$ but not in \WAL.

By a standard strategy in \cite{hansen2009neighbourhood}, we know that the basic polyadic modal logics ($\PML$) have the Craig Interpolation theorem. What's more, in \cite{santocanale2010uniform}, the authors proved that the minimal monotonic modal logic $M$ has Uniform Interpolation. Furtherly, we know that the basic modal logic $\SK$ also has Uniform Interpolation from \cite{andreka1995back} and \cite{andreka1998modal}. From the following three aspect we may conjecture that the basic $\WAL$ systems $\SK_{n}$ also has interpolation:
\begin{itemize}
    \item[1] $\WAL$ can be treated as a fragment of $\PML$.
    \item[2] $\SK_{n}$ is regarded as a general version of $\SK$, since $\SK$ is just $\SK_{1}$.
    \item[3] $\SK_{n}$ can be viewed as a special kind of monotonic modal logics.
    
\end{itemize}
But in fact no $\SK_{n}$ has the Craig Interpolation Property for $n \geq 2$. The first counterexample for interpolation we found is for $\SK_{3}$, which is relatively easier to understand and can be readily generalized to all $\SK_{n}$ for $n \geq 3$. Later we found a counterexample for $\SK_{2}$, which is slightly more complicated. Here we first give the two counter examples for $\SK_2$ and $\SK_3$ and then provide the general construction for $\SK_n$ ($n \ge 3$). But before we state the counterexamples, let us first clarify what do we mean by ``a counterexample'' of the Craig Interpolation Property for $\SK_n$.

\begin{lemma}
\label{lem:countercip}
Let $n$ be a non-zero natural number. If there are two pointed $n$-models $\M, w$ and $\N, v$ and two formulas $\phi$ and $\psi$ such that
\begin{enumerate}
    \item $\M, w\models \phi$ and $\N, v\models \psi$;
    \item $\SK_n \vdash \phi \to \lnot \psi$;
    \item letting $\Phi'$ be the set of all the propositional letters that appear both in $\phi$ and $\psi$, for any formula $\gamma$ in $\WAL$ such that only letters in $\Phi'$ appear, $\M, w \models \gamma$ iff $\N, v \models \gamma$;
\end{enumerate}
then $\SK_n$ lack the Craig Interpolation Property. 
\end{lemma}
\begin{proof}
Assume the antecedent and also that $\SK_n$ has the Craig Interpolation Property. Then since $\SK_n \vdash \phi \to \lnot \psi$, there is a interpolant $\gamma$ such that 
\begin{itemize}
    \item $\SK_n \vdash \phi \to \gamma$ and $\SK_n \vdash \gamma \to \lnot \psi$;
    \item only letters in $\Phi'$ appear in $\gamma$. 
\end{itemize}
Now since $\M, w \models \phi$ with $\M$ being an $n$-model and $\SK_n \vdash \phi \to \gamma$, by soundness, $\M, w \models \gamma$. Then $\N, v \models \gamma$ by the third bullet point in the antecedent. Then using $\SK_n \vdash \gamma \to \lnot \psi$ and soundness again, $\N, v \models \lnot \psi$, contradicting $\N, v \models \psi$. 
\end{proof}
Given this proposition, a pair of pointed $n$-models and a pair of formulas satisfying the antecedent constitute a counterexample of the Craig Interpolation property. Now we proceed to provide them for each $\SK_n$ with $n \ge 2$. 

% and we will give counterexamples for $\SK_{2}$ and $\SK_{3}$ here, from which one can easily find counterexamples for other $\SK_{i}$. The counterexamples are actually from what we lack in the bisimulation relation--which is not a bijection.

\begin{example} Consider the following two $2$-models where $\{\lr{w,w_1,w_2},\lr{w,w_3,w_4}\}$ is the ternary relation in the left model $\M_2$, and $\{\lr{v,v_1,v_2}\}$ is the ternary relation in the right model $\N_2$, where the valuations are as in the diagram. 
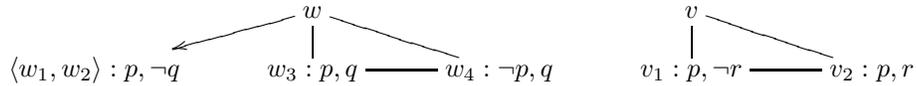
\vspace{-5pt}

$$\xymatrix@R-15pt@C+3pt{
 &w\ar@{-}[d]\ar[dl]\ar@{-}[dr] & &v\ar@{-}[d]\ar@{-}[rd] &\\
 \lr{w_1, w_2}:p,\neg q& w_3:p,q\ar@{-}[r] & w_4:\neg p,q &v_1: p,\neg r\ar@{-}[r] &v_2: p,r \\
} \vspace{-10pt}$$

\bigskip
Then set $\varphi_2 = \Box (\neg p \vee \neg q) \wedge \Diamond q$ and $\psi_2 = \Box (p \wedge r) \wedge \Box(p \wedge \neg r)$. It is easy to see that $\M_2, w \models \phi_2$ and $\N_2, v \models \psi_2$.  To see that $\SK_2\vdash \phi_2 \to \neg \psi_2$, consider the following derivation, where to make long Boolean combinations readable, we write negation of propositional letters as overline, omit $\land$ between purely Boolean formulas and replace $\lor$ with $|$.
\begin{itemize}
    \item $\vdash_2 \Box (\Bar{p} | \Bar{q}) \wedge \Box rp \wedge \Box \Bar{r}p \rightarrow \Box (((\Bar{p} | \Bar{q})rp) | ((\Bar{p} | \Bar{q})\Bar{r}p)| rp\Bar{r}p)$ \hfill $\AxWK_2$
    \item $\vdash_2 \Box (\Bar{p} | \Bar{q}) \wedge \Box rp \wedge \Box\Bar{r}p \rightarrow \Box p\Bar{q}$ \hfill PL,RE
    \item $\vdash_2 \phi_2 \wedge \psi_2 \rightarrow \Box p\Bar{q} \wedge \Diamond q$ \hfill PL
    \item $\vdash_2 \phi_2 \wedge \psi_2 \rightarrow \Box \Bar{q} \wedge \lnot \Box \Bar{q}$ \hfill PL, $\MonoK$
    \item $\vdash_2 \phi_2 \rightarrow \neg \psi_2$ \hfill PL
\end{itemize}
Here PL means propositional reasoning. Hence we are done with the first two points for this pair of models and formulas to be a counterexamples. For the last point, note that $Z=\{\lr{w,v}, \lr{w_1, v_1}, \lr{w_2, v_2}, \lr{w_3, v_1},\lr{w_3, v_2}\}$ is a $wa^2$-bisimulation when $\Phi = \{p\}$. Hence by Proposition \ref{prop:bisim-invariance}, for any formula $\gamma$ with $p$ the only propositional letter, $\M_2, w \models \gamma$ iff $\N_2, v \models \gamma$. But quite obviously, $p$ is the only common propositional letters in $\phi_2$ and $\psi_2$. Clearly now $\M, w$, $\N, v$, $\phi_2$, and $\psi_2$ form a counterexample to the Craig Interpolation Property for $\SK_2$. 
%\begin{align*}
%\M_2,w &\models \Box (\neg p \vee \neg q) \wedge \Diamond q(= \phi_2); \\
%\N_2,v &\models \Box (r \wedge p) \wedge \Box(\neg r \wedge p)(= \psi_2).
%\end{align*}
%\begin{align*}
%\M_2,w &\models \Diamond (p \wedge \neg q) \wedge \Box p\wedge \Diamond q \wedge \Box (\neg p \vee \neg q) (= \phi_2); \\
%\N_2,v &\models \Box (r \wedge p) \wedge \Box(\neg r \wedge p)\wedge \Box \neg p (= \psi_2).
%\end{align*}
\end{example}

\begin{example} Consider the following two $3$-models where $\{\lr{w,w_1,w_2,w_3}\}$ is the relation in $\M_3$ and $\{\lr{v,v_1,v_2,v_3}\}$ is the relation in $\N_3$. \vspace{-5pt}
$$\xymatrix@C+20pt@R-15pt{
&w_1:p,\neg q\ar@{-}[d] &     & v_1:\neg p,r\ar@{-}[d] \\
\M_3:w\ar@{-}[ru]\ar@{-}[rd]\ar@{-}[r]&w_2:p,q\ar@{-}[d] &  \N_3:v\ar@{-}[ru]\ar@{-}[rd]\ar@{-}[r]  & v_2:\neg p,\neg r \\
 & w_3:\neg p,q\ar@{-}[u]  &   & v_3:p,r\ar@{-}[u] & 
}$$
Then set $\phi_3= \Box p\Bar{q} \wedge \Box pq \wedge \Diamond(p|\Bar{p})$, $\psi_3= \Box \Bar{p}r \wedge \Box \Bar{p}\Bar{r} \wedge \Diamond(p|\Bar{p})$. Clearly $\M_3, w \models \phi_3$ and $\N_3, v \models \psi_3$. Further, $\SK_2 \vdash \phi_3 \to \lnot \psi_3$ since we have the following derivation. 
\begin{itemize}
    \item $\vdash_3 \Box p\Bar{q} \wedge \Box pq \wedge \Box\Bar{p}r \wedge \Box\Bar{p}\Bar{r} \rightarrow \Box (p\Bar{q}pq |  p\Bar{q}\Bar{p}r | p \Bar{q}\Bar{p}\Bar{r} | pq\Bar{p}r | pq\Bar{p}\Bar{r} | \Bar{p}r\Bar{p}\Bar{r}            )$ \hfill $\AxWK_3$
    \item $\vdash_3 \Box p\Bar{q} \wedge \Box pq \wedge \Box\Bar{p}r \wedge \Box\Bar{p}\Bar{r} \rightarrow \Box p\Bar{p}$ \hfill PL,
    \item $\vdash_3 \phi_3 \wedge \psi_3 \rightarrow \Box p\Bar{p} \wedge \Diamond (p|\Bar{p}) $ \hfill  PL, $\MonoK$
    \item $\vdash_3 \phi_3 \rightarrow \neg \psi_3$ \hfill PL
\end{itemize}
Finally, note that $Z=\{\lr{w,v}, \lr{w_1, v_3}, \lr{w_2, v_3}, \lr{w_3, v_1},\lr{w_3, v_2}\}$ is a $wa^3$-bisimulation if $\Phi = \{p\}$. 
\end{example}

%\begin{proposition} \label{lcit}
%$\SK_{2}$ and $\SK_{3}$ lack the Craig Interpolation Theorem.
%\end{proposition}
%\begin{proof}
%We will see that the above two are counterexamples for CIP: $\phi_i \rightarrow \neg \psi_i$ do %not have any interpolant in $\AxWK_i$.
%Since the only common proposition letter is $p$, we assume for contradiction that $\theta_i(p)$ %is a interpolant for $\phi_i \rightarrow \neg \psi_i$. Then We have $\vdash_i \phi_i \rightarrow %\theta_i(p)$ and $\vdash_i \theta_i(p) \rightarrow \neg \psi_3$. By the above two examples we %know that $\M_i,w\models\theta_i(p)$, which also means $\N_i,v\models\theta_i(p)$ by %$p$-bisimulation. But it follows that $\N_i,v\models \neg \psi_i$, a contradiction. 
%\end{proof}

The above example can be naturally generalized for each $\SK_n$ with $n > 3$. Let $m$ be the least natural number s.t. $2^{m} \geq n-1$ and pick $m$ many distinct propositional letters $r_1, \ldots, r_m$ from $\Phi$. Then for each $i$ from $1$ to $n-1$, we can associate a distinct conjunction of literals $\rho_i$ using $r_j$'s so that $\rho_i \land \rho_{i'}$ are incompatible for each $i \not= i'$. Then we can state the general counterexample.  

\begin{example} Consider the following two $n$-models where $\{\lr{w,w_1,...,w_n}\}$ is the relation in $\M_n$, and $\{\lr{v,v_1,...,v_n}\}$ is the relation in $\N_n$. 
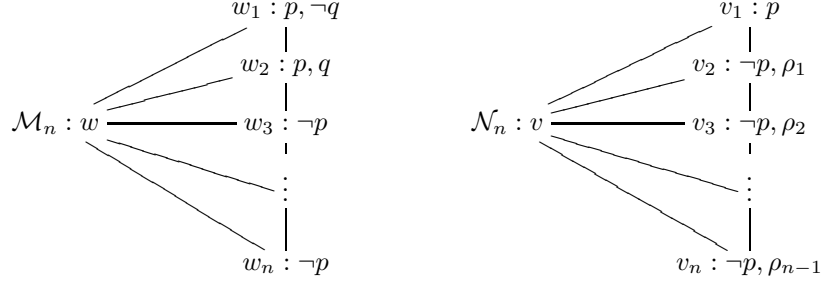
\vspace{-5pt}
$$\xymatrix@C+20pt@R-15pt{
&w_1:p,\neg q\ar@{-}[d] &     & v_1: p\ar@{-}[d] \\
 & w_2:p,q\ar@{-}[d] &  & v_2:\neg p,\rho_{1}\ar@{-}[d] \\
\M_n:w\ar@{-}[ru]\ar@{-}[rd]\ar@{-}[r]\ar@{-}[rdd]\ar@{-}[ruu]& w_3:\neg p\ar@{-}[d]  &  \N_n:v\ar@{-}[ru]\ar@{-}[rd]\ar@{-}[r]\ar@{-}[rdd]\ar@{-}[ruu]  & 
v_3:\neg p,\rho_{2}\ar@{-}[d] \\
  & \vdots &  & \vdots \\
 & w_n:\neg p \ar@{-}[u]  &   & v_n:\neg p,\rho_{n-1}\ar@{-}[u] &
}$$
Set $\phi_n =  \Box (p \wedge \neg q) \wedge \Box (p \wedge q) \wedge \Diamond \top$ and $\psi_n = \bigwedge_{i = 1}^{n-1}\Box (\neg p \wedge \rho_{i}) \wedge \Diamond \top$. Clearly $\M_n, w \models \phi_n$ and $\N_n, v \models \psi_n$. It is also easy to see that by $\mathtt{K}_n$, we can derive $(\Box(p \land \lnot q) \land \Box(p \land q) \land \bigwedge_{i = 1}^{n-1}\Box(\lnot p \land \rho_i)) \to \Box \bot$. With this we can then easily derive $\phi_n \to \lnot \psi_n$ in $\SK_n$. Finally, note that $p$ is the only common propositional letter in $\phi_n$ and $\psi_n$ and that $Z=\{\lr{w,v}, \lr{w_1, v_1}, \lr{w_2, v_1}\} \cup \{w_3,...,w_n\} \times \{v_2,...,v_n\} $ is a $wa^n$-bisimulation when $\Phi = \{p\}$. 
\end{example}

With the examples and Lemma \ref{lem:countercip}, the main theorem of this section follows. 
\begin{theorem}
For any $n \ge 2$, $\SK_n$ does not have the Craig Interpolation Property. 
\end{theorem}
\begin{remark}
Note that the Lemma \ref{lem:countercip} uses only the soundness of the logics. Hence for any extension of $\SK_n$ that is sound on $\M_n$ and $\N_n$, it still lacks the Craig Interpolation Property. For example, we may extend $\SK_n$ with $\mathtt{4}$ and our examples still work since $\mathtt{4}$ is valid on the underlying frames. 
%On the other hand, note how the example we provided requires a certain amount of distinct propositional letters. It is still open whether the Craig Interpolation Property can be found in $\SK_n$ when we restrict the cardinality of $\Phi$ to a suitably small finite number. 
\end{remark}

\section{Conclusion}\label{sec.con}

In this paper, we proved two results about $\WAL$: firstly, $\WAL$ have a van Benthem-Rosen characterization, and secondly, $\WAL$ do not have Craig Interpolation Property. We conclude with two potentially promising line of further investigation.
 
First, the main part of the completeness proof of $\SK_n$ over $n$-models is to solve some combinatorial puzzle \cite{nicholson2000revisiting}. Due to the semantics of $\WAL$ there is a natural link between combinatorics and $\WAL$, as also shown in the use of graph coloring problem in \cite{apostoli1995solution}. As future work, we would like to explore the possibility of using $\WAL$ to express interesting combinatorial properties in graph theory.

Second, even though we proved that that $\WAL$ do not have Craig Interpolation Theorem, it doesn't mean that the same must be the case with further constraints (stronger logics). For instance, the counterexample in our paper cannot show that $\SK_n \oplus \mathtt{T}$ lack CIT since the logic is not sound on the frames of the models we provided. What remains to be done then is to chart the map of CIP among the logics extending $\SK_n$'s and look for more general methods.

\bibliographystyle{splncs04}
\bibliography{myref}

\end{document}